\documentclass[notitlepage,floats,aps,nofootinbib,preprintnumbers,superscriptaddress,twocolumn,prl]{revtex4-1}

\usepackage{amsmath}
\usepackage{amsfonts,hyperref}
\usepackage{amssymb}
\usepackage{amsthm}
\usepackage{graphicx,float}
\usepackage{verbatim}
\usepackage{color}

\newtheorem{theorem}{Theorem}
\newtheorem{cor}{Corollary}
\newtheorem{Lemma}{Lemma}

\def\beq{\begin{equation*}\begin{aligned}}
\def\eeq{\end{aligned}\end{equation*}}
\def\beqn{\begin{equation}\begin{aligned}}
\def\eeqn{\end{aligned}\end{equation}}

\usepackage{amsthm}

\newtheorem{innercustomgeneric}{\customgenericname}
\providecommand{\customgenericname}{}
\newcommand{\newcustomtheorem}[2]{%
  \newenvironment{#1}[1]
  {%
   \renewcommand\customgenericname{#2}%
   \renewcommand\theinnercustomgeneric{##1}%
   \innercustomgeneric
  }
  {\endinnercustomgeneric}
}

\newcustomtheorem{customlemma}{Lemma}

\begin{document}

\title{All-Pay Auctions as Models for  Trade Wars and Military Annexation}
\author{Benjamin Kang,}
\affiliation{PRIMES, Massachusetts Institute of Technology, Cambridge, MA 02139 }
\author{James Unwin}
\affiliation{Department of Physics,  University of Illinois at Chicago, Chicago, IL 60607, USA}
\affiliation{Department of Physics, University of California, Berkeley \& Theoretical Physics Group, LBNL \& Mathematics Sciences Research Institute, Berkeley, CA 94720, USA}

\begin{abstract}
We explore an application of all-pay auctions to model trade wars and territorial annexation. Specifically, in the model we consider the expected resource, production, and aggressive (military/tariff) power are public  information, but actual resource levels are private knowledge. We consider the resource transfer at the end of such a competition which deprives the weaker country of some fraction of its original resources. In particular, we derive the quasi-equilibria strategies for two country conflicts under different scenarios. This work is relevant for the ongoing US-China trade war, and the recent Russian capture of Crimea, as well as historical and future conflicts.

\end{abstract}

\maketitle

\section{1.~Introduction}
\vspace{-1mm}

Auctions have long been used as mathematical models of competition \cite{Milgrom,Krishna,Capen,Cassady,Myerson,Vickrey,Wilson}. A setting useful for certain competitions is the `all pay' auction, in which losing bidders are required to pay a forfeit \cite{Kang,Krishna}.  International military conflicts provide perhaps the ultimate form of competition, and here we employ all pay auctions as a model of military territorial expansions and trade wars. We are unaware of any existing auction models of trade wars, although this is quite a natural description, and there are also few game theoretic analyses (although see e.g.~\cite{Harrison}).

Building on earlier work due to Hodler and Yekta\c s  \cite{Hodler} in which an all-pay auction model is used to model total military conquest,  we expand this framework to study the case of partial military conquest in which the victor takes only a fraction of the loser's resources, i.e.~the victor expands their territory at the expense of the loser. This generalization is  mathematically non-trivial and is more relevant to real world events, such as the Russian annexation of the Crimean Peninsula and the Israeli annexation of the Golan Heights. We also offer an reinterpretation of this framework in the context of trade wars.

This paper is structured as follows: In Section 2 we introduce the mathematical set-up of Hodler and Yekta\c s   \cite{Hodler}, and discuss how this model can be generalised to explore partial resource losses following a conflict between two countries. We formulate this scenario in terms of an all-pay auction with payoff function $W$. In Section~3 we discuss the mutual best responses for each country without incorporating constraints and then subsequently outline the conditions for quasi-equilibria. In contrast to \cite{Hodler}, obtaining the quasi-equilibria is mathematically  non-trivial and is the focus of Sections 4 \& 5. We first solve the relevant equations exactly for the special case that neither country has a competitive edge in Section~4. Then, in Section 5, we find an approximate solution for the quasi equilibrium strategies for the more general case. Section 6 provides a discussion of our results and an interpretation of the mathematical framework in terms of both conventional conflicts and trade wars.

\section{2.~Fractional Resource Forfeiture}
\vspace{-1mm}

	Emulating  \cite{Hodler},  we characterize the two countries (labelled by $i = 1, 2$) by their aggressive (e.g.~military)  power $\tilde{\lambda_i} \in \mathbb{R^+}$, production level $\tilde{\beta_i} \in \mathbb{R^+}$, and expected resource level $R_i \in \mathbb{R^+}$. While the values of these three variables are  public knowledge,  a private variable $r_i \in [0,1]$ is introduced such that the actual resource endowment of a country  $2r_iR_i$ is private information known only to a given country.  The private variables $r_i$ are publicly known to be distributed uniformly.

Furthermore, we introduce the notation $\lambda_i = 2R_1\tilde{\lambda_i}$ and $\beta_i = 2R_1\tilde{\beta_i}$ which corresponds to the maximum potential aggression  and production,  respectively. We also define the ratio between countries of these quantities
	\beq
	\lambda = \frac{\lambda_1}{\lambda_2} \qquad {\rm and} \qquad \beta = \frac{\beta_1}{\beta_2}.
	\eeq
	 It can be assumed without loss of generality that $\beta \leq \lambda$; implying Country $1$ has an advantage in aggressive potential, while Country $2$ has an advantage in production.
	 
	Suppose that $b_i$ is the resource allocation, such that $b_i\lambda_i$ is allocated to aggressive actions and $(r_i-b_i)\beta_i$ is allocated to production and define the equilibrium strategies $f_i(r_i) = b_i$  such that $f_i(r_i) :[0,1] \rightarrow [0,r_i]$. It will be useful to note the follow lemma from \cite{Hodler}:

\vspace{0.1mm}
	\begin{Lemma}[HY \cite{Hodler}] For a function $f_i(r_i)\in[0,r_i]$, in any monotone equilibrium then $f_1(0) =f_2(0) = 0$. Further, $f_1$ and $f_2$ are non-decreasing and  $\lambda f_1(1) = f_2(1)$.
\label{l1}\end{Lemma}

\vspace{-3mm}
It is assumed that the country with more aggressive potential will win the competition, acquiring a fraction of resources  $\alpha \in [0,1] $ from the losing party, while also retaining all of their initial resource endowment. This can be encapsulated in the following payoff function
\beqn W_i = \begin{cases}
	\alpha\beta_i(r_i-b_i) & \lambda_ib_i < \lambda_jb_j\\[3pt]
	\beta_i(r_i-b_i)  & \lambda_ib_i = \lambda_jb_j \\[3pt]
	\beta_i(r_i-b_i) + (1-\alpha)\beta_j(r_j-b_j) & \lambda_ib_i > \lambda_jb_j
\end{cases}.\label{WWW}
\eeqn
Note that taking $\alpha=1$ recovers the case studied in  \cite{Hodler}.

\section{3.~Equilibrium Strategies}
\vspace{-1mm}

Given the form of $W_i$ it follows that the payoff for Country $1$ making an optimal bid $y=f(r_1)$ is given by 
\beq
u_1(y)=&\int_0^{f_2^{-1}(\lambda y)}\left[\beta_1(r_1-y)+(1-\alpha)\beta_2(r_2-f_2(r_2))\right]{\rm d}r_2\\
&+ \int_{f_2^{-1}(\lambda y)}^1\left[\alpha\beta_1(r_1-y)\right]{\rm d}r_2.
\eeq
One can formulate each countries mutual best response:

\begin{Lemma}
Disregarding any constraints, the mutual best responses for each party is to follow strategies of the form
\beqn \label{eq:8}
f_1(r_1) = &\frac{K_0}{(\lambda+2\beta)(1-\alpha)}((1-\alpha)r_1+\alpha)^{\frac{\beta}{\lambda}}\\
&+\frac{\beta}{\beta+2\lambda}r_1-\frac{\alpha(\beta\lambda + \beta+2\lambda)}{(\beta+\lambda)(\beta+2\lambda)(1-\alpha)} \\ &+ K_1((1-\alpha)r_1+\alpha)^{-(1+\frac{\beta}{\lambda})},\\[5pt]
f_2(r_2) = &\frac{\lambda\beta K_0^{-\frac{\lambda}{\beta}}}{(2\lambda+\beta)(1-\alpha)}((1-\alpha)r_2+\alpha)^{\frac{\lambda}{\beta}}\\
&+\frac{\lambda}{2\beta+\lambda}r_2-\frac{\alpha\beta\lambda(1+ 2\beta+\lambda)}{(\beta+\lambda)(2\beta+\lambda)(1-\alpha)} \\ &+ K_2((1-\alpha)r_2+\alpha)^{-(1+\frac{\lambda}{\beta})}~,
\eeqn
where $K_0$, $K_1$, and $K_2$ are arbitrary constants.
\end{Lemma}
\begin{proof}
To identify the value of $y$ which maximizes the payoff function $u_1(y)$ for Country 1, we set the derivative with respect to $y$ to zero to obtain the condition
\beqn \label{eq:1} 
F(y)\frac{df_2^{-1}(\lambda y)}{dy}
=\beta f_2^{-1}(\lambda y) + \frac{\alpha\beta}{1-\alpha}.&
\eeqn
where $F(y)=\beta(f_1^{-1}(y)-y)+f_2^{-1}(\lambda y)-\lambda y$.
Writing Country $2$'s bid similarly as $y = f_2(r_2)$ to maximise the payoff $y$ should satisfy
\beqn \label{eq:2} 
F(y)\frac{df_1^{-1}(y)}{\lambda dy}
=f_1^{-1}(y) +  \frac{\alpha}{1-\alpha}.&
\eeqn
Since the factor $F(y)$ appears in both eqns.~\eqref{eq:1} \& \eqref{eq:2} we can substitute for one of these factors, leading to 
\beq
 (\ln[(1-\alpha)f_2^{-1}(\lambda y) + \alpha])' = \frac{\beta}{\lambda}(\ln[(1-\alpha)f_1^{-1}(y) + \alpha])'.
\eeq
Integrating, it follows that
\beqn \label{xd}(1-\alpha)f_2^{-1}(\lambda y) + \alpha = K_0((1-\alpha)f_1^{-1}(y) + \alpha)^{\frac{\beta}{\lambda}}~,
\eeqn
	where $K_0$ is a constant. To proceed we identify that Country 1's optimal bid $f_1(r_1) = y$ is related to Country 2's optimal bid $f_2(r_2) =z$ by  $y = z/\lambda$. 
	
Then substituting eq.~\eqref{xd} back into eq.~\eqref{eq:1} and eq.~\eqref{eq:2}, and also substituting for $x$ and $y$, we obtain the following two equations:
\beq \label{eq:6}
\lambda \frac{df_1(r_1)}{{\rm d}r_1} = &K_0((1-\alpha)r_1+\alpha)^{\frac{\beta}{\lambda}-1}-\frac{\alpha}{(1-\alpha)r_1 + \alpha}\\ & - \frac{(1-\alpha)(\beta+\lambda)f_1(r_1)}{(1-\alpha)r_1 + \alpha} +\frac{(1-\alpha)\beta r_1}{(1-\alpha)r_1 + \alpha},\\[5pt]
\beta \frac{df_2(r_2)}{{\rm d}r_2} = &\lambda\beta K_0^{-\frac{\lambda}{\beta}}((1-\alpha)r_2+\alpha)^{\frac{\lambda}{\beta}-1}-\frac{\lambda\beta\alpha}{(1-\alpha)r_2 + \alpha} \\ &- \frac{(1-\alpha)(\beta+\lambda)f_2(r_2)}{(1-\alpha)r_2 + \alpha} +\frac{(1-\alpha)\lambda r_2}{(1-\alpha)r_2 + \alpha}.
\eeq
Solving the above first order differential equations for $f_1$ and $f_2$ yields the stated result.
\end{proof}

To find the quasi-equilibrium strategies of each player, we apply the  boundary conditions below (following \cite{Hodler})
\beqn
& f_1(0)=f_2(0)=0\\
 &
 f_1(1)=f_2(1)/\lambda~.
\label{bcbc}
\eeqn
These conditions are appropriate to the problem due to the properties highlighted in Lemma \ref{l1}.

 These boundary conditions lead to three equations involving three unknown constants $K_0,K_1,K_2$.
Solving the two conditions at $r_i=0$ for $K_1$ and $K_2$ in terms of $K_0$, the final boundary condition  $f_1(1)=f_2(1)/\lambda$ can be expressed in the following form
\beqn \label{eq:13}
a K_0+b K_0^{-\frac{\lambda}{\beta}}
& =c,
\eeqn
where 
\beqn\label{abc}
a&=\frac{\lambda (1 - \alpha^{1 + 2\frac{\beta}{\lambda}})}{(\lambda+2\beta)(1-\alpha)}~, \\
b&= \frac{\lambda\beta (1 - \alpha^{1 + 2\frac{\lambda}{\beta}})}{(2\lambda+\beta)(\alpha-1)}~,
 \\
c&=-\frac{\beta\lambda}{\beta+2\lambda}+\frac{\alpha\lambda(\beta\lambda + \beta+2\lambda)(1 - \alpha^{1 + \frac{\beta}{\lambda}})}{(\beta+\lambda)(\beta+2\lambda)(1-\alpha)} \\  
&~~~~+\frac{\lambda}{2\beta+\lambda}
-\frac{\alpha\beta\lambda(1+ 2\beta+\lambda)(1 - \alpha^{1 + \frac{\lambda}{\beta}})}{(\beta+\lambda)(2\beta+\lambda)(1-\alpha)}~.
\eeqn
Unfortunately, eq.~\eqref{eq:13} does not a have an exact general solution, but as a first approach in the next section we will investigate the case  $\lambda = \beta$ which is exactly solvable.

\section{4.~Quasi-equilibria for the case $\lambda = \beta$ }
\vspace{-1mm}

	The case  $\lambda = \beta$ implies $\lambda_1/\beta_1=\lambda_2/\beta_2$ and represents the scenario in which neither country has a comparative advantage in aggressive power for a given level of  production.  Note that  $\lambda = \beta=1$ would imply that both countries are equal in both aggressive power and production level, and thus neither country has any advantage. 
	
		\begin{theorem}
	The quasi-equilibrium strategies for each country for $\lambda = \beta$ and $r_i \in [0,1]$ are
	\beq
	&f_1(r_1) = \frac{\beta + 1}{3\beta}r_1 -\frac{\alpha(\beta + 1)}{6\beta(1-\alpha)}+ \frac{\alpha^3(\beta+1)}{6\beta(1-\alpha)^3(r_1+\frac{\alpha}{1-\alpha})^{2}},\\[5pt]
 &f_2(r_2) = 
	\frac{\beta +1}{3}r_2-\frac{\alpha(\beta+1)}{6(1-\alpha)}+ \frac{\alpha^3(\beta+1)}{6(1-\alpha)^3(r_2+\frac{\alpha}{1-\alpha})^{2}}.
	\eeq
	\end{theorem}
	\begin{proof}
		For $\lambda = \beta$ the exponent of $K_0^{-\frac{\lambda}{\beta}}$ is simply $-1$ and the coefficients of eq.~\eqref{abc} simplify to 
	\beq
a&=\frac{ (1 - \alpha^3)}{3(1-\alpha)}~, 
\hspace{15mm}
b= \frac{\beta (1 - \alpha^3)}{3(\alpha-1)}~,
 \\
c&=\frac{1-\beta}{3}+\frac{1}{3}\alpha(1+\alpha)(1-\beta)~.
\eeq
Therefore, eq.~\eqref{eq:13} simplifies to 
\beq \label{eq:14}
\left( \frac{1+\alpha+\alpha^2}{3}\right)\left(K_0^2 -(1-\beta)K_0 - \beta\right) =0,
	\eeq
which has solutions  $K_0 = 1, -\beta$.
Since  $\beta\geq0$ and from eq.~\eqref{xd} it follows $K_0\geq 0$, we hence take $K_0 = 1$. Then evaluating the forms of $K_1$ and $K_2$ (obtained through the conditions $f_1(0) = f_2(0) = 0$) we obtain
\beq \label{eq:15}
	K_1 = \frac{\alpha^3(\beta+1)}{6\beta(1-\alpha)},
\qquad
	K_2 = \frac{\alpha^3(\beta+1)}{6(1-\alpha)}.
	\eeq
	Substituting $K_0=1,$ and the forms for $K_1$ and $K_2$ above, into the expressions of Lemma \ref{l1} completes the proof.
		\end{proof}

			\begin{cor}
	The quasi-equilibrium strategies for each country for $\lambda = \beta$ are related via
\beq
	 f_2(r_i)=\beta f_1(r_i)~.
	 \eeq
	 			\end{cor}

	This result implies that whichever country has the absolute advantage in both aggressive power and production is expected to win the trade war. Moreover, for the sub-case $\lambda = \beta=1,$ in which the two countries have equal strength in both aggressive power and production, the  equilibrium strategies are identical.

\section{5.~Quasi-equilibria for the case $\lambda \neq \beta$}
\vspace{-1mm}

An analysis of $\lambda \neq \beta$ requires a little care, since for $\lambda\neq\beta$ the equations exhibit singularities for $r_i\rightarrow0$. As a result, a minor modification is required to consider the more general case.
Specifically, to obtain the equilibrium strategies for the more general case we will restrict the private resource variable to $r_i \in [\epsilon,1]$ for some small positive constant $\epsilon$. The presence of a fixed $\epsilon$ cuts off the divergences and in this case eq.~(\ref{eq:13}) can be solved. The analysis however is substantially more complicated than the earlier case studied here (and the study in \cite{Hodler}). 
	
	Introducing this infinitesimal $\epsilon$, we adapt the boundary conditions of eq.~(\ref{bcbc}) to the following
	\beqn
	 f_1(\epsilon)=f_2(\epsilon)=0, \hspace{5mm}
	f_1(1)=f_2(1)/\lambda~.
	\eeqn
	Taking eq.~(\ref{eq:8}) of Lemma 2 and 
	evaluating $f_1$ at $\epsilon$ gives 
	\beq \label{eq:11}
	f_1(\epsilon)&=\frac{K_0}{(\lambda+2\beta)(1-\alpha)}\left[\alpha\left(1+\frac{1-\alpha}{\alpha}\epsilon\right)\right]^{\frac{\beta}{\lambda}} + \frac{\beta \epsilon}{\beta+2\lambda}\\
	&-\frac{\alpha(\beta\lambda + \beta+2\lambda)}{(\beta+\lambda)(\beta+2\lambda)(1-\alpha)} + K_1\left[\alpha(1+\frac{1-\alpha}{\alpha}\epsilon)\right]^{-(1+\frac{\beta}{\lambda})}.
	\eeq
	Taking first the boundary  condition $f_1(\epsilon)=0$ and solving for $K_1$ at first order in the infinitesimal quantity $\epsilon$ implies
\beqn
K_1=\frac{ \frac{K_0\alpha^{\frac{\beta}{\lambda}}}{(\lambda+2\beta)(1-\alpha)}\left(1+\frac{(1-\alpha)\beta\epsilon}{\alpha\lambda}\right) + \frac{\beta\epsilon}{\beta+2\lambda} -\frac{\alpha(\beta\lambda + \beta+2\lambda)}{(\beta+\lambda)(\beta+2\lambda)(1-\alpha)}}{\alpha^{-(1+\frac{\beta}{\lambda})}((1+\frac{\beta}{\lambda})\frac{1-\alpha}{\alpha}\epsilon-1)}.
\label{10}
\eeqn
	Then applying the next condition $f_2(\epsilon)=0$ and solving for $K_2$ at first order in  $\epsilon$ gives
\beqn
K_2=\frac{ \frac{\lambda\beta K_0^{-\frac{\lambda}{\beta}}\alpha^{\frac{\lambda}{\beta}}}{(2\lambda+\beta)(1-\alpha)}\left(1+\frac{(1-\alpha)\lambda\epsilon}{\alpha\beta}\right) + \frac{\lambda\epsilon}{2\beta+\lambda}-\frac{\alpha\beta\lambda(1+ 2\beta+\lambda)}{(\beta+\lambda)(2\beta+\lambda)(1-\alpha)}}{\alpha^{-(1+\frac{\lambda}{\beta})}((1+\frac{\lambda}{\beta})\frac{1-\alpha}{\alpha}\epsilon-1)}.
\label{11}
\eeqn
	The final boundary condition $f_1(1)=f_2(1)/\lambda$ requires 
	\beqn 
&\frac{K_0}{(\lambda+2\beta)(1-\alpha)} +\frac{\beta}{\beta+2\lambda}-\frac{\alpha(\beta\lambda + \beta+2\lambda)}{(\beta+\lambda)(\beta+2\lambda)(1-\alpha)} + K_1\\[5pt]
 = &\frac{\beta K_0^{-\frac{\lambda}{\beta}}}{(2\lambda+\beta)(1-\alpha)} +\frac{1}{2\beta+\lambda}-\frac{\alpha\beta(1+ 2\beta+\lambda)}{(\beta+\lambda)(2\beta+\lambda)(1-\alpha)} + \frac{K_2}{\lambda}.
\eeqn
Substituting the forms of $K_1$ and $K_2$ into the above and rearranging, we arrive at the condition
\beqn\label{xx}
&
  K_0 =\frac{C_3- C_1 K_0^{-\frac{\lambda}{\beta}} }{C_2}~,
\eeqn
where $C_1$, $C_2$ and $C_3$ are parameter dependent constants
\beq
C_1&=  \frac{\beta }{(1-\alpha)(2\lambda+\beta)} \left(1+ \frac{\alpha^{\frac{\lambda}{\beta}}\left(1+\frac{\lambda(1-\alpha)}{\alpha\beta}\epsilon\right)}{\alpha^{-(1+\frac{\lambda}{\beta})}((1+\frac{\lambda}{\beta})\frac{1-\alpha}{\alpha}\epsilon-1)}\right),\\
C_2&=\frac{1}{(1-\alpha)(\lambda+2\beta)}\left(1 + \frac{\alpha^{\frac{\beta}{\lambda}}\left(1+\frac{\beta(1-\alpha)}{\alpha\lambda}\epsilon\right)}{\alpha^{-(1+\frac{\beta}{\lambda})}((1+\frac{\beta}{\lambda})\frac{1-\alpha}{\alpha}\epsilon-1)}\right),\\
C_3&=\frac{\alpha}{(\beta+\lambda)(1-\alpha)}\left(\frac{\beta(1+ 2\beta+\lambda)}{(2\beta+\lambda)} -\frac{(\beta\lambda + \beta+2\lambda)}{(\beta+2\lambda)}\right) \\
&+\frac{\beta}{\beta+2\lambda} + \frac{\frac{\beta}{\beta+2\lambda}\epsilon -\frac{\alpha(\beta\lambda + \beta+2\lambda)}{(\beta+\lambda)(\beta+2\lambda)(1-\alpha)}}{\alpha^{-(1+\frac{\beta}{\lambda})}((1+\frac{\beta}{\lambda})\frac{1-\alpha}{\alpha}\epsilon-1)}\\
& -\frac{1}{2\beta+\lambda} - \frac{ \frac{1}{2\beta+\lambda}\epsilon-\frac{\alpha\beta(1+ 2\beta+\lambda)}{(\beta+\lambda)(2\beta+\lambda)(1-\alpha)}}{\alpha^{-(1+\frac{\lambda}{\beta})}((1+\frac{\lambda}{\beta})\frac{1-\alpha}{\alpha}\epsilon-1)}.
 \eeq
Solving eq.~(\ref{xx}) analytically for $K_0$ is challenging, however an approximate solution can be found via iterative methods, as we show below. 

To proceed we take the zeroth order approximation to be $K_0^{(0)}=1$ (the value found in Theorem 1). The first order value $K_0^{(1)}$ is found by taking  $K_0^{(0)}\approx1$ on the RHS of eq.~(\ref{xx}), implying
$K_0^{(1)}\approx\frac{C_3- C_1}{C_2}$.
Then an approximate solution for $K_0$ is provided by the second order iteration $K_0^{(2)}$ found by 
substituting $K_0^{(1)}$ into the RHS of eq.~(\ref{xx}) to obtain
\beqn
K_0^{(2)}\approx \frac{C_3}{C_2} - \frac{C_1}{C_2} \left(\frac{C_3- C_1}{C_2}\right)^{-\frac{\lambda}{\beta}}~.
\eeqn

\pagebreak
Finally, using this second order approximate expression $K_0^{(2)}$, along with the exact forms of $K_1$ and $K_2$ obtained in eqns.~(\ref{10}) \& (\ref{11}) in eq.~(\ref{eq:8}) of Lemma 2, we derive the approximate equilibrium strategies for the more general case with $\lambda \neq \beta$. Specifically, we obtain the following result:

	\begin{theorem}
	The quasi-equilibrium strategies for each country for $r_i \in [\epsilon,1]$ to leading order in $\epsilon$ and to second order precision in $K_0$ are given by eq.~(\ref{yy}).

\vspace{2mm}

\begin{widetext}
\beqn \label{yy}
f_1(r_1) \approx & \left( \frac{C_3}{C_2} - \frac{C_1}{C_2} \left(\frac{C_3- C_1}{C_2}\right)^{-\frac{\lambda}{\beta}}\right)\frac{((1-\alpha)r_1+\alpha)^{\frac{\beta}{\lambda}}}{(\lambda+2\beta)(1-\alpha)}
+\frac{\beta}{\beta+2\lambda}r_1-\frac{\alpha(\beta\lambda + \beta+2\lambda)}{(\beta+\lambda)(\beta+2\lambda)(1-\alpha)} \\ &+ \Bigg[\frac{ \frac{K_0}{(\lambda+2\beta)(1-\alpha)}\alpha^{\frac{\beta}{\lambda}}(1+\frac{\beta(1-\alpha)}{\alpha\lambda}\epsilon) + \frac{\beta}{\beta+2\lambda}\epsilon -\frac{\alpha(\beta\lambda + \beta+2\lambda)}{(\beta+\lambda)(\beta+2\lambda)(1-\alpha)}}{\alpha^{-(1+\frac{\beta}{\lambda})}((1+\frac{\beta}{\lambda})\frac{1-\alpha}{\alpha}\epsilon-1)}\Bigg]((1-\alpha)r_1+\alpha)^{-(1+\frac{\beta}{\lambda})},\\[10pt]
f_2(r_2)  \approx &\lambda\beta \left( \frac{C_3}{C_2} - \frac{C_1}{C_2} \left(\frac{C_3- C_1}{C_2}\right)^{-\frac{\lambda}{\beta}}\right)^{-\frac{\lambda}{\beta}} \frac{((1-\alpha)r_2+\alpha)^{\frac{\lambda}{\beta}}}{(2\lambda+\beta)(1-\alpha)}
+\frac{\lambda}{2\beta+\lambda}r_2-\frac{\alpha\beta\lambda(1+ 2\beta+\lambda)}{(\beta+\lambda)(2\beta+\lambda)(1-\alpha)} \\ &+ \Bigg[\frac{ \frac{\lambda\beta K_0^{-\frac{\lambda}{\beta}}}{(2\lambda+\beta)(1-\alpha)}\alpha^{\frac{\lambda}{\beta}}(1+\frac{\lambda(1-\alpha)}{\alpha\beta}\epsilon) + \frac{\lambda}{2\beta+\lambda}\epsilon-\frac{\alpha\beta\lambda(1+ 2\beta+\lambda)}{(\beta+\lambda)(2\beta+\lambda)(1-\alpha)}}{\alpha^{-(1+\frac{\lambda}{\beta})}((1+\frac{\lambda}{\beta})\frac{1-\alpha}{\alpha}\epsilon-1)}\Bigg]((1-\alpha)r_2+\alpha)^{-(1+\frac{\lambda}{\beta})}~.
\eeqn
\end{widetext}
	\end{theorem}

\vspace{-1mm}
\section{6.~Discussion}
\vspace{-3mm}

Auction theory has previously been used to model competition between animals \cite{Bishop}, countries  \cite{ONeill,Hodler} and individuals or groups  \cite{Baye,Siegel,Snyder}. Moreover, for game theoretic approaches to military conflicts see \cite{Haavelmo,Garfinkel,Hirshleifer,Skaperdas,Grossman}. Here we have extended the earlier work of Hodler and Yekta\c s \cite{Hodler},  on forfeiture due to conflicts, to include the important case of partial resource losses. In our analysis the fraction of resources ceded by the loser is controlled by the variable $\alpha\in [0,1] $. Note that the variable $\alpha$ measures loss of resources, thus in the context of military conquest $\alpha$ need not strictly measure captured land area, since resources (oil, minerals, etc.) may not be distributed uniformly.

Notably, one of the more novel and topical applications of the all pay auction model developed here is to model trade wars. Trade wars occur when two countries create tariffs in response to the other country. Each country is looking to limit their opponent economically, however this also typically harms their domestic consumers, as the price of certain items will increase due to lack of imported materials and thus creation of tariffs entails a cost.  
The original interpretation put forward in \cite{Hodler} was that $\tilde{\lambda_i}$ corresponded to military power and the forfeit for the losing country was a total loss of resources. The framework presented here  can be reinterpreted as a model of trade wars if the variable  $\tilde{\lambda_i}$ is taken to be the ability of a country to introduce tariffs, while $\alpha$ is interpreted as an economic loss (e.g.~market share) relative to competitors. In the context of trade wars, Lemma 3 gives the exact equilibrium strategies when neither country has a comparative advantage  for a given level of production, and Lemma 6 gives the approximate solution for countries with dissimilar tariffs or production rates. 

We note that this analysis assumes net-zero economic loss between the two countries. This simplifying assumption could be relaxed by modifying the $(1 - \alpha)$ term in eq.~(\ref{WWW}) to include an additional $\kappa\in[0,1]$ variable which accounts for loss to third party competitors who are not otherwise involved in the trade war, such that $\alpha\in[0,\kappa]$.

		Trade wars have become an important element of current world economics. However, there are  few theoretical analyses of trade wars and this work, in the context of auction theory, adds a new aspect to the literature.

\vspace{2mm}
{\bf Acknowledgements.} 
This research was undertaken as part of the MIT-PRIMES program. JU is grateful for support from the National Science Foundation through grant  DMS-1440140 while at the Mathematics Sciences Research Institute, Berkeley during Fall 2019.

\newpage

\begingroup
\renewcommand{\section}[2]{}%

\endgroup

\end{document}